\documentclass{article}
\usepackage{a4wide}

\usepackage{times,authblk}
\usepackage[UKenglish]{babel}
\usepackage{amsmath, amsfonts, amssymb, amsthm, bm, stmaryrd, enumitem, mathtools}
\usepackage{graphicx, tikz-cd}

\renewcommand{\epsilon}{\varepsilon}
\renewcommand{\phi}{\varphi}
\renewcommand{\rho}{\varrho}

\newtheorem{Def}{Definition}[section]
\newenvironment{definition}{\begin{Def} \rm}{\end{Def}}
\newtheorem{lemma}[Def]{Lemma}
\newtheorem{proposition}[Def]{Proposition}

\newtheorem{theorem}[Def]{Theorem}
\newtheorem{example}[Def]{Example}

\newtheorem*{theorem*}{Theorem}

\newcommand{\komma}{,\hspace{0.3em}}

\renewcommand{\leq}{\leqslant}
\renewcommand{\geq}{\geqslant}
\renewcommand{\emptyset}{\varnothing}

\newcommand{\Reals}{{\mathbb R}}
\newcommand{\Complexes}{{\mathbb C}}
\newcommand{\Quaternions}{{\mathbb H}}

\newcommand{\notperp}{\mathbin{\not\perp}}

\renewcommand{\c}{^\perp}
\newcommand{\cc}{^{\perp\perp}}

\newcommand{\herm}[2]{\left( #1 , #2 \right)}
\newcommand{\lin}[1]{[#1]}

 \DeclareMathOperator{\down}{\downarrow}

\begin{document}

\title{A characterisation of orthomodular spaces \\ by Sasaki maps}

\author{Bert Lindenhovius}

\author{Thomas Vetterlein}

\affil{\footnotesize
Department of Knowledge-Based Mathematical Systems,
Johannes Kepler University Linz \authorcr
Altenberger Stra\ss{}e 69, 4040 Linz, Austria}

\date{\today}

\maketitle

\begin{abstract}\parindent0pt\parskip1ex
\noindent Given a Hilbert space $H$, the set $P(H)$ of one-dimensional subspaces of $H$ becomes an orthoset when equipped with the orthogonality relation $\perp$ induced by the inner product on $H$. Here, an \emph{orthoset} is a pair $(X,\perp)$ of a set $X$ and a symmetric, irreflexive binary relation $\perp$ on $X$. In this contribution, we investigate what conditions on an orthoset $(X,\perp)$ are sufficient to conclude that the orthoset is isomorphic to $(P(H),\perp)$ for some orthomodular space $H$, where \emph{orthomodular spaces} are linear spaces that generalize Hilbert spaces. In order to achieve this goal, we introduce \emph{Sasaki maps} on orthosets, which are strongly related to Sasaki projections on orthomodular lattices. We show that any orthoset $(X,\perp)$ with sufficiently many Sasaki maps is isomorphic to $(P(H),\perp)$ for some orthomodular space, and we give more conditions on $(X,\perp)$ to assure that $H$ is actually a Hilbert space over $\mathbb R$, $\mathbb C$ or $\mathbb H$.



\mbox{}\vspace{-2ex}
\end{abstract}

\section{Introduction}\label{sec:introduction}

In order to obtain a deeper understanding of quantum physics, it is imperative to analyse the structure of complex Hilbert spaces, the objects that are used to represent quantum systems mathematically. For the following, it is useful to dissect the definition of a Hilbert space as follows. Firstly, we call a division ring equipped with an involutorial antiautomorphism $(-)^\star$ a $\star$-\emph{sfield}. Then an \emph{(anisotropic) Hermitian space} is a linear space over a $\star$-sfield $K$ that is equipped with an anisotropic, symmetric sesquilinear form $(-,-):H\times H\to K$. In case $H$ is a Hermitian space over one of the classical $\star$-sfields, that is, $\mathbb R$, $\mathbb C$, or $\mathbb H$, and if in addition $H$ is complete with respect to the norm induced by the inner product, we call $H$ a \emph{(classical) Hilbert space}.

One way to analyse the structure of Hilbert spaces that we aim at following, consists of trying to find different axioms for Hilbert spaces in terms of simpler structures. For example, the closed subspaces of Hilbert spaces form a complete orthomodular lattice, hence one could choose orthomodular structures as a starting point, and try to add conditions to assure that an orthomodular lattice is isomorphic to the lattice of closed subspaces of a Hilbert space. 

Since we aim at disentangling the structure of Hilbert spaces as much as possible, we choose an even simpler structure. Our starting point is the orthogonality relation $\perp$ on a Hilbert space, or more generally, a Hermitian space $H$: we say that $x,y\in H$ are \emph{orthogonal} and write $x\perp y$ if $\herm x y = 0$. 
The relation $\perp$ induces a symmetric and irreflexive binary relation on the set $P(H)$ of one-dimensional subspaces of $H$, representing the pure states of the quantum system modelled by $H$. This induced relation, which we also denote by $\perp$, is called the \emph{orthogonality relation} and keeps track of transitions between pure states that have a probability of zero. The pair $(P(H),\perp)$ is perhaps the bare minimum of a possible structure necessary for quantum physics, and Foulis proposed to generalize this structure as follows:

\begin{definition}
	An \emph{orthoset} (sometimes also called an \emph{orthogonality space}) is a set $X$ equipped with a symmetric and irreflexive binary relation $\perp$, called the \emph{orthogonality relation}. We call $x,y\in X$ \emph{orthogonal} if $x\perp y$. The \emph{rank} of $X$ is the supremum of the cardinalities of sets of mutually orthogonal elements.
\end{definition}

We refer to \cite{Dac,Wil} for more details on orthosets, and discuss here only the concepts we need for the following. The definition of an orthoset is essentially that of an undirected graph. However, only a few graph-theoretic concepts are relevant for our purposes. For instance a \emph{morphism of orthosets} is simply defined as a graph homomorphism. Moreover, a subset of mutually orthogonal elements of an orthoset, also called a \emph{$\perp$-set}, is precisely a clique of the orthoset regarded as a graph. Then the rank of the orthoset is the supremum of the cardinalities of its cliques. 

The main motivation for considering orthosets doesn't lie in graph theory, but in our example of the one-dimensional subspaces of a Hilbert spaces, which can be generalized to Hermitian spaces as follows:

\begin{example}\label{ex:hilbert space orthoset}
	Let $H$ be a Hermitian space over a $\star$-sfield $K$. For each nonzero $x\in H$, let $\lin x$ be the linear span of $x$. Then the relation $\perp$ on $H$ induces an orthogonality relation on $P(H) = \{ \lin x \colon x \in H \setminus \{0\} \}$, the projective space associated with $H$: $\lin x \perp \lin y$ means $x \perp y$. Obviously, $(P(H), \perp)$ is an orthoset.
\end{example}

We note that also other quantum structures give rise to orthosets, for instance the non-empty closed subspaces of a Hilbert space, representing the propositions of the quantum system represented by $H$. This example can be generalized to \emph{orthocomplemented} posets, i.e., bounded posets $P$ equipped with an order-reversing involution $(-)^\perp:P\to P$ such that $x\wedge x^\perp=0$ for each $x\in P$. 
\begin{example}\label{ex:OMP orthoset}
	Let $P$ be an orthocomplemented poset. Let $X=P\setminus\{0\}$. Let $\perp$ be the relation on $X$ defined by $x\perp y$ if and only if $x\leq y^\perp$. Then $(X,\perp)$ is an orthoset.
\end{example}

Since we view orthosets as abstractions of pure state spaces, we are less interested in the orthosets in Example \ref{ex:OMP orthoset}. Instead, we are interested in finding conditions that assure that a given orthoset $(X,\perp)$ is isomorphic to the orthoset in Example \ref{ex:hilbert space orthoset} for some Hilbert space $H$. In particular, this means we aim at achieving the following three goals: 
\begin{itemize}
	\item[(i)] finding conditions on $(X,\perp)$ that assure that $X=P(H)$ for some Hermitian space $H$ over a $\star$-sfield $K$;
	\item[(ii)] finding conditions on $(X,\perp)$ that imply that $K=\mathbb C$;
	\item[(iii)] finding conditions on $(X,\perp)$ that guarantee that $H$ is metrically complete with respect to the norm on $H$ induced by the form $(-,-)$.
\end{itemize}

The last goal can be reformulated as follows. Given a Hermitian space $H$, for any subset $S\subseteq H$ let $S^\perp$ be the set of all $x\in X$ such that $x\perp y$ for each $y\in S$. Then $S^\perp$ is a subspace of $H$ and any subspace of this form is called  \emph{closed}. We say that $H$ is \emph{orthomodular} if $M+M^\perp=H$ for each closed subspace $M$ of $H$. Now, a Hermitian space $H$ over one of the classical $\star$-sfields, i.e., $\mathbb C$, $\mathbb R$ or $\mathbb H$, is orthomodular if and only if $H$ is metrically complete \cite{AmAr}, so (iii) can be replaced by:  
\begin{itemize}
	\item[(iii')] finding conditions on $(X,\perp)$ that guarantee that $H$ is an orthomodular space. \end{itemize}
We note that (i) and (iii') can be combined into:
\begin{itemize}
	\item[(i')] finding conditions on $(X,\perp)$ that assure that $X=P(H)$ for some orthomodular space $H$ over a $\star$-sfield $K$. \end{itemize}

This contribution deals with goal (i'). Under the assumption that (i') has been achieved, (ii) already has been dealt with in a satisfactory way by means of automorphisms, which is a natural choice of structure, since for any Hilbert space $H$, automorphisms on $(P(H),\perp)$ are related to unitary transformations on $H$, which describe the symmetries of the physical system represented by $H$ such as time evolution. See also Section \ref{sec:Classical Hilbert spaces}.


With respect to goal (i'), it has been known for a long time how to characterize orthomodular spaces by lattice-theoretic means. However, the involved properties do not have a straightforward physical meaning. Moreover, this requires associating lattices to orthosets, whereas we prefer to formulate our conditions directly in the framework of orthosets.

More explicitly, given $A\subseteq X$, we can define the set $A^\perp:=\{x\in X:x\perp a\text{ for each }a\in A\}$. We call $A\subseteq X$ \emph{orthoclosed} if $A=A^{\perp\perp}$. 
Given a Hermitian space $H$, any closed subspace $S$ of $H$ is a Hermitian space itself, hence the notation $P(S)$ makes sense. Then all orthoclosed subsets of $(P(H),\perp)$ are of the form $P(S)$ for some closed subspace $S$ of $H$.

The set $\mathcal C(X,\perp)$ of all orthoclosed subsets of an orthoset $(X,\perp)$ is a complete ortholattice if we order it by inclusion and with $(-)^\perp$ as orthocomplementation. Conversely, given an atomistic ortholattice $L$ with orthocomplementation $(-)^\perp$, the set $\mathrm{At}(L)$ of atoms of $L$ becomes an orthoset when equipped with the relation $\perp$ defined by $a\perp b$ if and only if $a\leq b^\perp$. The next lemma states that these operations yield a one-to-one correspondence between some class of ortholattices and some class of orthosets:

\begin{lemma}\cite[Proposition 2]{Vet3}\label{lem:pointclosed-completeatomistic}
	Let $(X,\perp)$ be a \emph{point-closed} orthoset, i.e., every singleton subset of $X$ is orthoclosed. Then $L:=\mathcal C(X,\perp)$ is a complete atomistic ortholattice, and the map $X\to\mathrm{At}(L)$, $x\mapsto\{x\}$ is an isomorphism of orthosets.
	
	Conversely, given a complete atomistic ortholattice $L$, let $X:=\mathrm{At}(L)$. Then $(X,\perp)$ is a point-closed orthoset, and the map $L\to \mathcal C(X,\perp)$, $p\mapsto\{x\in X:x\leq p\}$ is an isomorphism of ortholattices.
\end{lemma}

Because we are interested in models of pure state spaces, we will restrict ourselves to point-closed orthosets.


The crucial lattice-theoretic condition one usually imposes on $\mathcal C(X,\perp)$ is the \emph{covering property}, which an ortholattice $L$ possesses if for each $x\in L$ and atom $a\in L$ such that $a\wedge x=0$ (or equivalently, such that $a\nleq x$), we have that $x\vee a$ covers $x$. This is a property that has to be stated in the framework of ortholattices instead of in that of orthosets. Moreover, at first sight it is not clear what the physical interpretation of the covering property is.

We aim at solving all these issues by introducing so-called \emph{Sasaki maps} on orthosets. These maps are closely related to Sasaki projections, which explains our choice of terminology, and are natural to use because their definition reminds us to that of a projection operator, since the first defining condition expresses idempotency, and the second self-adjointness. 

For the next definition, we introduce the notation $\complement S$ to denote the set-theoretic complement of a subset $S$ of a set $X$.

\begin{definition}
	Let $A$ be an orthoclosed subset of an orthoset $(X,\perp)$. A map
	\[ \phi \colon \complement A\c \to A \]
	is called a \emph{Sasaki map} to $A$ if the following conditions hold:
	\begin{itemize}
		
		\item[\rm (S1)] $\phi(e) = e$ for all $e \in A$.
		
		\item[\rm (S2)] For any $e, f \in\complement A^\perp$, we have $\phi(e) \perp f \text{ if and only if } e \perp \phi(f)$.
	\end{itemize}
	We call $(X,\perp)$ a \emph{Sasaki space} if for any orthoclosed subset $A$ of $X$, there exists a Sasaki map $\varphi_A$ to $A$.
\end{definition}

\begin{definition}
An orthoset $(X,\perp)$ is called \emph{reducible} if $X$ is the disjoint union of two non-empty sets $A$ and $B$ such that $e \perp f$ for any $e \in A$ and $f \in B$. If $(X,\perp)$ is not reducible, we call it \emph{irreducible}.
\end{definition}
We note that if $(X,\perp)$ is a reducible, and can be written as the union of $A$ and $B$ as defined above, it is still possible that $A$ contains distinct mutually orthogonal elements. Hence, reducibility is not the same as stating that $(X,\perp)$ is a bipartite graph. However, there is an illuminating graph-theoretic description of irreducibility. We call a graph $(V,E)$ \emph{connected} if for each pair of vertices $x,y\in V$, there is a finite path from $x$ to $y$, i.e., there are $v_1,v_2,\ldots,v_n$ such that $x=v_1$, $y=v_n$ and $v_i Ev_{i+1}$ for $i=1,2,\ldots,n-1$. If we regard an orthoset $(X,\perp)$ as a graph, then its complement graph is precisely $(X,\notperp)$. 

\begin{lemma}
	An orthoset $(X,\perp)$ is irreducible if and only if $(X,\not\perp)$ is connected when regarded as a graph.
\end{lemma}
\begin{proof}
	Assume $(X,\perp)$ is reducible, so $X$ is the disjoint union of non-empty subsets $A$ and $B$ such that $a\perp b$ for each $a\in A$ and each $b\in B$. Fix $a\in A$ and $b\in B$. If $(X,\not\perp)$ were connected, there would be a path $v_1,\ldots,v_n$ from $a$ to $b$, so $a=v_1$, $b=v_n$ and $v_i\not\perp v_{i+1}$ for each $i=1,\ldots,n-1$. Now, there must be some $i\in\{1,\ldots,n-1\}$ such that $v_i\in A$ and $v_{i+1}\in B$. Hence, we obtain both $v_i\perp v_{i+1}$ since $(X,\perp)$ is reducible, and $v_i\notperp v_{i+1}$ since $v_1,\ldots,v_n$ is a path in $(X,\not\perp)$, which is clearly a contradiction. So $(X,\not\perp)$ must be disconnected. 
	
	Conversely, assume $(X,\not\perp)$ is disconnected. Then $(X,\notperp)$ has at least two non-empty connected components. Let $A$ be one of these components, and let $B$ be its set-theoretic complement in $X$. Let $a\in A$ and $b\in B$. Then there is no path from $a$ to $b$ in $(X,\perp)$, and in particular, we do not have $a\notperp b$. Thus, $a\perp b$, showing that $(X,\perp)$ is reducible.
	\end{proof}

We are now able to formulate our main theorem, which implies that  Sasaki maps do not only provide a way to incorporate the covering property within the framework of orthosets, but they also provide a condition that is strong enough to assure that an orthoset is induced by a Hermitian space, even an orthomodular space.

\begin{theorem}
	Let $(X,\perp)$ be an irreducible, point-closed Sasaki space of rank $\geq 4$. Then $(X,\perp)$ is isomorphic to $(P(H),\perp)$ for some orthomodular space $H$.
	
	Under this isomorphism, the Sasaki maps are, for any closed subspace $S$ of $H$ and for any $x\notperp S$, given by
	\[ \phi_{P(S)}(\lin x) \;=\; (\lin x + S\c) \cap S. \]
\end{theorem}
The latter expression for $\phi_{P(S)}$ indeed shows the relation between Sasaki maps and Sasaki projections. 

We give an outline for the rest of the article. In Section \ref{sec:examples}, we give some examples of Sasaki spaces. In Section \ref{sec:orthomodular spaces}, we describe the relation between Sasaki spaces and orthomodular spaces, culminating in the proof of our main theorem. In Section \ref{sec:Sasaki projections}, we show that Sasaki maps of a Sasaki space $(X,\perp)$ induce a full Sasaki set of projections on $\mathcal C(X,\perp)$ in the sense of Finch. Finally, in the last section we give a characterization of infinite-dimensional classical Hilbert spaces in terms of Sasaki spaces.

\section{Examples of Sasaki spaces}\label{sec:examples}

In this section, we explore some examples of Sasaki spaces. The most important example is provided by our guiding example:

\begin{proposition}
	Let $H$ be an orthomodular space. Then $(P(H),\perp)$ is a Sasaki space.
\end{proposition}

\begin{proof}
	Since $H$ is orthomodular, it follows that for any closed subspace $S$ of $H$ the orthogonal projection $P_S$ associated to $S$ exists. This map $P_S$ is an idempotent, self-adjoint linear operator. For any $x \notperp S$, we define $\phi(\lin x) = \lin{P_S(x)}$. Then $\phi$ is a map from $\complement P(S)\c$ to $P(S)$ with the required properties, where we recall that $\complement P(S)^\perp$ is the set-theoretic complement of $P(S)^\perp$ in $P(H)$.
\end{proof}

Before we identify some examples of Sasaki spaces that are not associated to some orthomodular space, we need one lemma.

\begin{lemma}\label{lem:sasakimap}
	Let $(X,\perp)$ be an orthoset and let $A$ be an orthoclosed subset of $X$. Then each of the following conditions is sufficient to assure the existence of a Sasaki map for $A$:
	\begin{itemize}
		\item[(a)] $A^\perp=\complement A$;
		\item[(b)] $\complement A^\perp$ is a $\perp$-set;
		\item[(c)] $A$ is a singleton.
	\end{itemize}
\end{lemma}
\begin{proof}
	For (a), we note that $A^\perp=\complement A$ implies $\complement A^\perp=A$, whence we can take the identity on $A$. For (b), if $\complement A^\perp$ is a $\perp$-set, then $A^\perp=\complement A$, hence this case reduces to (a). Finally, given (c), i.e.,  if $A$ is a singleton, say $A=\{a\}$ for some $a\in X$, then the only possible map $\varphi:\complement A^\perp\to A$ is the map that is constant $a$. Since $\varphi$ clearly satisfies (S1), we only have to verify (S2). For each $e\in\complement A^\perp$, we have $e\notin A^\perp$, so $e\not\perp a$. Hence for each $e,f\in\complement A^\perp$, we have $\varphi(e)=a\not\perp f$ and $e\not\perp a=\varphi(f)$, so $\varphi$ is a Sasaki map.
\end{proof}

\begin{example}
	Let $X$ be a set and let $\perp$ be the relation $\neq$, or equivalently, let $(X,\perp)$ be the complete undirected graph on $|X|$ vertices. Then ${\mathcal C}(X,\perp)={\mathcal P}(X)$, the power set of $X$. Moreover, for each $A\in{\mathcal C}(X,\perp)$, we have $A^\perp=\complement A$, whence $(X,\perp)$ is a Sasaki space by (a) of Lemma \ref{lem:sasakimap}.
\end{example}

\begin{example}
	Let $(X,\perp)$ be the orthoset with $X=\{a,b,c,d\}$ with $a\perp b\perp c\perp d\perp a$, corresponding to the cyclic undirected graph on four vertices as depicted below:
	\begin{center}
		\begin{tikzpicture}
			\node (a) at (0,0) {$a$};
			\node [right  of=a] (b)  {$b$};
			\node [below  of=a] (d)  {$d$};
			\node [below  of=b] (c)  {$c$};
			\draw [black,  thick] (a) -- (b);
			\draw [black,  thick] (c) -- (b);
			\draw [black,  thick] (c) -- (d);
			\draw [black,  thick] (a) -- (d);
		\end{tikzpicture}
	\end{center}
	Then ${\mathcal C}(X,\perp)=\{X,\emptyset,\{a,c\},\{b,d\}\}$, where $\{a,c\}^\perp=\{b,d\}$ and $\{b,d\}^\perp=\{a,c\}$. For any $A\in{\mathcal C}(X,\perp)$, we have $\complement A^\perp=A$, hence $(X,\perp)$ is a Sasaki space by (a) of Lemma \ref{lem:sasakimap}.
\end{example}

\begin{example}
	Let $(X,\perp)$ be the orthoset with $X=\{a,b,c,d\}$, and where the orthogonality relation $\perp$ is specified by $a\perp b$ and $c\perp d$, so $(X,\perp)$ corresponds to the undirected graph below:
	\begin{center}
		\begin{tikzpicture}
			\node (a) at (0,0) {$a$};
			\node [right  of=a] (b)  {$b$};
			\node [below  of=a] (d)  {$d$};
			\node [below  of=b] (c)  {$c$};
			\draw [black,  thick] (a) -- (b);
			\draw [black,  thick] (c) -- (d);
		\end{tikzpicture}
	\end{center}
	Then $\{a\}^\perp=\{b\}$, and $\{c\}^\perp=\{d\}$. For any $S\subseteq X$ consisting of more than one point we have $S^\perp=\emptyset$. Hence ${\mathcal C}(X,\perp)=\{X,\emptyset,\{a\},\{b\},\{c\},\{d\}\}$. Since $\complement X=\emptyset=X^\perp$, there exist Sasaki maps for both $X$ and $\emptyset$ using Lemma \ref{lem:sasakimap}.(a). All other sets in ${\mathcal C}(X,\perp)$ are singletons, whence $(X,\perp)$ is a Sasaki space by (c) of the same lemma.
\end{example}

The name `Sasaki map' suggests a connection with Sasaki projections on an orthomodular lattice. We include a definition: 
\begin{definition}\label{def:Sasaki projection}
	Let $L$ be an orthomodular lattice. Then the map $\pi_x:L\to L$, $y\mapsto x\wedge(x^\perp\vee y)$ is called the \emph{Sasaki projection} to $x\in L$.
\end{definition}

Indeed, in Proposition \ref{prop:OMLinducesSasakispace} below we will show that any Sasaki projection on a complete orthomodular lattice induces a Sasaki map on an orthoset associated to the lattice as in Example \ref{ex:OMP orthoset}. Before we show this, we need the following facts about Sasaki projections:

\begin{lemma}\label{lem:Sasaki projection facts}
	Let $L$ be an orthomodular lattice, and let $x\in L$. Then for each $y,z\in L$ we have
	\begin{itemize}
		\item[(a)] $y\leq x$ if and only if $\pi_x(y)=y$;
		\item[(b)] $\pi_x(\pi_x(y^\perp)^\perp))\leq y$;
		\item[(c)] $\pi_x(y)=0$ if and only if $y\leq x^\perp$;
		\item[(d)] $\pi_x(y)\perp z$ if and only if $y\perp \pi_x(z)$.
	\end{itemize}
\end{lemma}
\begin{proof}
	Properties (a) and (b) follow from the paragraph between Lemma 5 and Theorem 6 in \cite{Fou60}. Property (c) follows from \cite[Lemma 1]{Fou62}. For (d), assume that $\pi_x(y)\perp z$, i.e.,  $z\leq\pi_x(y)^\perp$. Since $\pi_x$ is clearly monotone, we have $\pi_x(z)\leq \pi_x(\pi_x(y)^\perp))$, which implies $\pi_x(z)\leq y^\perp$ by (b). Thus, $y\perp\pi_x(z)$. By interchanging $y$ and $z$ in the obtained implication and using the symmetry of $\perp$ twice, we obtain the implication in the other direction.
\end{proof}

\begin{proposition}\label{prop:OMLinducesSasakispace}
	Let $L$ be a complete orthomodular lattice, and let $(X,\perp)$ be the orthoset obtained from $L$ as in Example \ref{ex:OMP orthoset}, so $X=L\setminus\{0\}$. Then $(X,\perp)$ is a Sasaki space with ${\mathcal C}(X,\perp)=\{X\cap\down x:x\in L\}$, which is isomorphic to $L$ as an orthomodular lattice. For any $A=X\cap\down x$ in $\mathcal C(X,\perp)$, restricting and corestricting the Sasaki projection $\pi_{x}:L\to L$ to a map $\varphi_A:\complement A^\perp\to A$ yields a Sasaki map to $A$.
\end{proposition}
\begin{proof}
	Since $0$ is the only element in $L$ that is orthogonal to itself, it is clear that $(X,\perp)$ is an orthoset. Let $x\in X$. Then \[\{x\}^\perp=\{y\in X:y\perp x\}=\{y\in X:y\leq x^\perp\}=\down x^\perp\cap X.\] Moreover, for each non-empty $S\subseteq X$, we have 
	\begin{align*}
		S^\perp & =\{x\in X:x\perp s\text{ for all }s\in S\}\\
		& =\{x\in X:x\leq s^\perp\text{ for all }s\in S\}=\left(\down \bigwedge_{s\in S}s^\perp\right)\cap X.
	\end{align*}
	Since any $A\in\mathcal C(X,\perp)$ is of the form $S^\perp$ for some $S\subseteq X$, it follows that each $A\in\mathcal C(X,\perp)$ is of the form $X\cap\down x$ for some $x\in L$. It is easy to verify that the resulting bijection $\psi:L\to\mathcal C(X,\perp)$, $x\mapsto X\cap\down x$ is an isomorphism of orthomodular lattices.
	
	Now for $A\in\mathcal C(X,\perp)$, with $A=X\cap\down x$ for some $x\in L$, let $\varphi_A:\complement A^\perp\to A$ be the restriction and the corestriction of the Sasaki projection $\pi_{x}:L\to L$. We show that $\varphi_A$ is well defined by showing that $\pi_{x}(y)\neq 0$ for each $y\in\complement A^\perp$. Indeed, we have $\complement A^\perp=\complement(X\cap\down x)\c=\complement\psi(x)\c=\complement\psi(x\c)=X\setminus(X\cap \down x^\perp)=X\setminus\down x\c$. It now follows from (c) of Lemma \ref{lem:Sasaki projection facts} that for each $y\in X$, we have $\pi_{x}(y)\neq 0$ if and only if $y\nleq x\c$, i.e., if and only if $y\in\complement A^\perp$. From (a) and (d) of the same lemma follows that $\varphi_A$ is a Sasaki map.
\end{proof}

\section{Orthomodular spaces}\label{sec:orthomodular spaces}
In this section, we aim at characterising orthomodular spaces by means of their associated Sasaki spaces.

We start by showing that the definition of a Sasaki space can be weakened slightly. We first need one lemma.
\begin{lemma}\label{lem:perp-set contained in orthoclosed set}
	Let $(X,\perp)$ be an orthoset, let $A$ be an orthoclosed subspace of $X$, and let $D$ be a maximal $\perp$-set contained in $A$. If there exists a Sasaki map $\varphi_{D^\perp}$ to $D\c$, then $A=D\cc$.
\end{lemma}
\begin{proof}
	Since $D\subseteq A$, and $A$ is orthoclosed, we have $D^{\perp\perp}\subseteq A$. Assume that there is some $e \in A \setminus D\cc$. Then $e \notperp D\c$, and since $A^\perp\subseteq D^\perp$, we have $x \notperp D\c$ for each $x \in A\c$. Hence $e \perp x = \phi_{D\c}(x)$ implies $\phi_{D\c}(e) \perp x$ for any $x \in A\c$. This means that $\phi_{D\c}(e) \in A\cc=A$, and since $\phi_{D\c}(e) \perp D$ by definition of a Sasaki map for $D^\perp$, we obtain a contradiction with the maximality of $D$.
\end{proof}

\begin{proposition}
	An orthoset $(X,\perp)$ is a Sasaki space if and only if for each $\perp$-set $D\subseteq X$ there is a Sasaki map to $D^\perp$. 
\end{proposition}
\begin{proof}
	The `only if' direction is trivial. For the other direction, let $A$ be an orthoclosed subset of $X$ and let $D$ be a maximal $\perp$-set contained in $A\c$. Assume that there exists a Sasaki map $\varphi_{D\c}$ to $D\c$. It follows from Lemma \ref{lem:perp-set contained in orthoclosed set} that $A\c=D\cc$. Since $A$ is orthoclosed, we obtain $A=A\cc=D^{\perp\perp\perp}=D\c$, which implies that $\varphi_{D\c}$ is a Sasaki map to $A$.
\end{proof}

An orthoset $(X,\perp)$ is called a \emph{Dacey space} if its associated ortholattice ${\mathcal C}(X,\perp)$ is orthomodular. The next example shows that not every orthoset is a Dacey space.

\begin{example}\label{ex:non-dacey space}
	Let $(X,\perp)$ orthoset with $X=\{a,b,c,d\}$ and where the orthogonality relation $\perp$ is specified by $a\perp b\perp c\perp d$, corresponding to the undirected graph below.
	
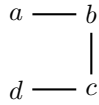
\begin{center}
		\begin{tikzpicture}
			\node (a) at (0,0) {$a$};
			\node [right  of=a] (b)  {$b$};
			\node [below  of=a] (d)  {$d$};
			\node [below  of=b] (c)  {$c$};
			\draw [black,  thick] (a) -- (b);
			\draw [black,  thick] (c) -- (b);
			\draw [black,  thick] (c) -- (d);
		\end{tikzpicture}
	\end{center}
	
	Then \[{\mathcal C}(X,\perp)=\{X,\emptyset,\{b\},\{c\},\{a,c\},\{b,d\}\},\] where $\{b\}^\perp=\{a,c\}$ and $\{c\}^\perp=\{b,d\}$. Thus ${\mathcal C}(X,\perp)$ is the benzene ring, which is a standard example of an ortholattice that is not orthomodular. Hence $(X,\perp)$ is not a Dacey space.
\end{example}

The following useful criterion for this property is due to Dacey \cite{Dac}.

\begin{lemma} \label{lem:Dacey}
	${\mathcal C}(X, \perp)$ is orthomodular if and only if, for any $A \in {\mathcal C}(X, \perp)$ and any maximal $\perp$-set $D \subseteq A$, we have $A = D\cc$.
\end{lemma}

\begin{theorem} \label{thm:Sasaki-is-Dacey}
	Any Sasaki space is a Dacey space.
\end{theorem}

\begin{proof}
	Let $D$ be a maximal $\perp$-set contained in an orthoclosed subset $A$ of $X$. Since $D^\perp$ is orthoclosed, there exists a Sasaki map $\varphi_{D\c}$ to $D\c$, hence $A=D\cc$ by Lemma \ref{lem:perp-set contained in orthoclosed set}. Thus, it follows from Lemma \ref{lem:Dacey} that $(X,\perp)$ is indeed a Dacey space.
\end{proof}

From the last theorem it follows that the orthoset in 
Example \ref{ex:non-dacey space} cannot be a Sasaki space, since it is not even a Dacey space.

Next, we present an example of a Dacey space that is not a Sasaki space.

\begin{example}
	Let $A=\{0,a,a^\perp, 1\}$ be a Boolean algebra of four elements, and let $B=\{0,b,b\c,c,c\c,d,d\c,1\}$ be a Boolean algebra of eight elements with atoms $b,c,d$. Let $L$ be the horizontal sum of $A$ and $B$ regarded as orthomodular lattices as depicted in the Hasse diagram below:
	
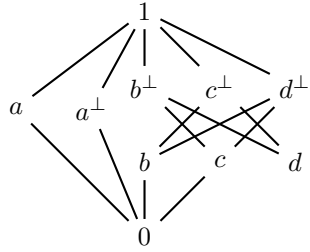
\begin{center}
		\begin{tikzpicture}
			\node (a) at (0,0) {$a$};
			\node (app) at (1,0) {$\phantom{a.}$};
			\node (ap) at (1,0.08)  {$a^\perp$};
			\node [below right of=ap] (b)  {$b$};
			\node [right  of=b] (c)  {$c$};
			\node [right  of=c] (d)  {$d$};
			\node [above  of=b] (bp)  {$b^\perp$};
			\node [above  of=c] (cp)  {$c^\perp$};
			\node [above  of=d] (dp)  {$d^\perp$};
			\node [above  of=bp] (1)  {$1$};
			\node [below  of=b] (0)  {$0$};
			\draw [black,  thick] (a) -- (1);
			\draw [black,  thick] (a) -- (0);
			\draw [black,  thick] (ap) -- (1);
			\draw [black,  thick] (ap) -- (0);
			\draw [black,  thick] (cp) -- (1);
			\draw [black,  thick] (c) -- (0);
						\draw [black,  thick] (d) -- (0);
			\draw [black,  thick] (bp) -- (1);
			\draw [black,  thick] (b) -- (0);
			\draw [black,  thick] (dp) -- (1);
			\draw [black,  thick] (b) -- (cp);
			\draw [black,  thick] (b) -- (dp);
			\draw [black,  thick] (c) -- (bp);
			\draw [black,  thick] (c) -- (dp);
			\draw [black,  thick] (d) -- (bp);
			\draw [black,  thick] (d) -- (cp);
		\end{tikzpicture}
	\end{center}
	
	Then $L$ is an orthomodular lattice without the covering property: we have $a\wedge b=0$, but $a\vee b=1$, which does not cover $b$. Let $X=\{a,a^\perp,b,c,d\}$, be the set of atoms of $L$. Let $(X,\perp)$ be the orthoset obtained by Lemma \ref{lem:pointclosed-completeatomistic} corresponding to the undirected graph depicted below:
	
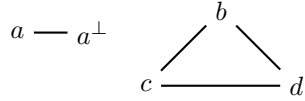
\begin{center}
		\begin{tikzpicture}
			\node (a) at (0,0) {$a$};
			\node (app) at (1,0) {$\phantom{a.}$};
			\node (ap) at (1,0.08)  {$a^\perp$};
			\node [below right  of=ap] (c)  {$c$};
			\node[right of=c](phantom1){$\phantom{a}$};
			\node[above  of=phantom1](phantom2){$\phantom{a}$};
			\node[above of=phantom1](b){$b$};
			\node [right of=phantom1] (d)  {$d$};
			\draw [black,  thick] (a) -- (app);
			\draw [black,  thick] (c) -- (b);
			\draw [black,  thick] (c) -- (d);
			\draw [black,  thick] (b) -- (d);
		\end{tikzpicture}
	\end{center}
	
	For $x\in\{a,a^\perp\}$, we have $\{x\}^\perp=\{x^\perp\}$. For $x\in\{b,c,d\}$, we have $\{x\}^\perp=\{b,c,d\}\setminus\{x\}$. For distinct $x,y\in X$, we have $\{x,y\}^\perp=\{b,c,d\}\setminus\{x,y\}$ if $x,y\in\{b,c,d\}$. In all other cases we have $\{x,y\}^\perp=\emptyset$. 
	Hence \[\mathcal C(X,\perp)=\big\{\emptyset,\{a\},\{a^\perp\},\{b\},\{c\},\{d\},\{b,c\},\{b,d\},\{c,d\},X\big\}.\] Then $\psi:L\to\mathcal C(X,\perp)$ given by $0\mapsto \emptyset$, $1\mapsto X$, $x\mapsto\{x\}$ for $x\in X$, and $x^\perp\mapsto\{b,c,d\}\setminus\{x\}$ for $x=b,c,d$ is clearly an isomorphism of orthomodular lattices by Lemma \ref{lem:pointclosed-completeatomistic}. Thus $(X,\perp)$ is a Dacey space. However, consider $A=\{c,d\}$ in $\mathcal C(X,\perp)$. Then $\complement A^\perp=\complement \{b\}=\{a,a^\perp,c,d\}$. If a Sasaki map $\varphi:\complement A^\perp\to A$ exists, we must have $\varphi(c)=c$ and $\varphi(d)=d$ by condition (S1) of a Sasaki map. For $a\in \complement A^\perp$, we have $\varphi(a)\in A=\{c,d\}$. Without loss of generality assume $\varphi(a)=c$. Then $\varphi(a)=c\perp d$, whereas $a\not\perp d=\varphi(d)$. Hence $\varphi$ cannot be a Sasaki map, so $(X,\perp)$ is a Dacey space that is not a Sasaki space.  
\end{example}

Next we investigate what conditions a Dacey space has to satisfy in order to be a Sasaki space, for which we need a lemma, which was stated in \cite[4.4]{Wil} without a proof. For the convenience of the reader, we include a proof. We note that with \emph{basic} elements of a bounded poset, we mean elements that are either atoms or the least element of the poset. This concept was already introduced in \cite{HHLN}, and it allows for a more elegant formulation of the lemma. 

\begin{lemma}\label{lem:Wilce}
	An orthomodular lattice $L$ has the covering property if and only if all its Sasaki projections send basic elements to basic elements.
\end{lemma}

\begin{proof}
	Assume $L$ has the covering property. Let $x\in L$ and let $a\in L$ be an atom. Assume first that $a\leq x^\perp$. Then $\pi_x(a)=0$ by (c) of Lemma \ref{lem:Sasaki projection facts}. Now assume that $a\nleq x^\perp$. By orthomodularity of $L$, we have that $x^\perp\vee a=x^\perp\vee y$ for some $y\in L$ such that $x^\perp\perp y$. By the covering property, $x^\perp$ is covered by $x^\perp\vee a=x^\perp\vee y$, whence $y$ is an atom of $L$, and $\pi_x(a)=x\wedge(x^\perp \vee a)=x\wedge(x^\perp\vee y)=y$. 
	
	Conversely, assume that all Sasaki projections on $L$ send basic elements to basic elements. Let $a\in\mathrm{At}(L)$ and $x\in L$ such that $a\nleq x$. Since $L$ is orthomodular, there is a $y\in L$ such that $x\perp y$ and $a\vee x=x\vee y$. Note that $y\neq 0$, because otherwise $a\vee x=x$ contradicting that $a\nleq x$. Then $\pi_{x^\perp}(a)=x^\perp\wedge(x\vee a)=x^\perp\wedge(x\vee y)=y$ by orthomodularity, which we can apply since $x\perp y$. Thus $y$ must be a basic element of $L$, so an atom of $L$. Since $y$ is an atom that is orthogonal to $x$, it follows that $x$ is covered by $x\vee y=x\vee a$. 
\end{proof}

\begin{lemma}\label{lem:point closed Dacey with covering property implies Sasaki space}
	Let $(X,\perp)$ be a point-closed Dacey space such that ${\mathcal C}(X,\perp)$ has the covering property. Then $(X,\perp)$ is a Sasaki space.
\end{lemma}
\begin{proof}
	Since $(X,\perp)$ is point-closed, $\{x\}$ is an element of ${\mathcal C}(X,\perp)$ for each $x\in X$. Since $(X,\perp)$ is a Dacey space, ${\mathcal C}(X,\perp)$ is an orthomodular lattice. Let $A\in{\mathcal C}(X,\perp)$. Since ${\mathcal C}(X,\perp)$ has the covering property, it follows from Lemma \ref{lem:Wilce} that the Sasaki projection $\pi_A$ sends each atom of ${\mathcal C}(X,\perp)$ to either the least element or another atom of ${\mathcal C}(X,\perp)$, i.e., $\pi_A(\{x\})$ is either empty or a singleton for each $x\in X$. By (c) of Lemma \ref{lem:Sasaki projection facts} we have $\pi_A(\{x\})=\emptyset$ if and only if $\{x\}\subseteq A^\perp$, hence $\pi_A(\{x\})\neq\emptyset$ if and only if $x\in\complement A^\perp$. By the covering property it then follows then that there is a unique $\varphi(x)\in A$ such that $\pi_A(\{x\})=\{\varphi(x)\}$, which defines a map $\varphi:\complement A\c\to A$. Let $x\in A$. By (a) of Lemma \ref{lem:Sasaki projection facts} we have $\pi_A(\{x\})=\{x\}$, so $\varphi(x)=x$. Let $x,y\in\complement A\c$. Using (d) of the same lemma, $\varphi(x)\perp y$ if and only if $\pi_A(\{x\})=\{\varphi(x)\}\subseteq\{y\}^\perp$ if and only if $\pi_A(\{x\})\perp\{y\}$ if and only if $\{x\}\perp\pi_A(\{y\})=\{\varphi(y)\}$ if and only if $x\perp\varphi(y)$. Thus $\varphi$ is a Sasaki map.
\end{proof}

\begin{lemma} \label{lem:Sasaki-and-join-with-atom}
	Let $(X,\perp)$ be a Sasaki space. Then for each $A \in {\mathcal C}(X,\perp)$ and each $e \notin A$, we have
$A \vee \{ e \}\cc = A \vee \{ \phi_{A\c}(e) \}\cc.$
\end{lemma}

\begin{proof}
	Let $x \perp A$. Then $e, x \notin A$, hence $e$ and $x$ are in the domain of $\phi_{A\c}$. Since $\phi_{A\c}(x) = x$, we have that $x \perp e$ is equivalent to $x \perp \phi_{A\c}(e)$. We conclude that $(A \cup \{e\})\c = (A \cup \{\phi_{A\c}(e)\})\c$. The assertion now follows from the fact that $(S\cup T)\cc=S\cc\vee T\cc$ for each $S,T\subseteq X$.
\end{proof}
In the following, an \emph{AC lattice} is meant to be an atomistic lattice with the covering property. 
\begin{lemma} \label{lem:Sasaki-is-AC}
	Let $(X,\perp)$ be a point-closed Sasaki space. Then ${\mathcal C}(X,\perp)$ is AC.
\end{lemma}

\begin{proof}
	Since $X$ is point-closed, ${\mathcal C}(X,\perp)$ is atomistic, the atoms being $\{e\}$, $e \in X$.
	
	To show the covering property, let $A \in {\mathcal C}(X,\perp)$ and let $e \in X$ be such that $\{ e \} \nsubseteq A$, that is, $e \notin A$. By Lemma \ref{lem:Sasaki-and-join-with-atom}, $A \vee \{ e \} = A \vee \{ \phi_{A\c}(e) \}$. By Theorem \ref{thm:Sasaki-is-Dacey}, ${\mathcal C}(X,\perp)$ is orthomodular and since $\{ \phi_{A\c}(e) \}$ is an atom, we conclude that $A \vee \{e\}$ covers $A$.
\end{proof}

We collect our results, and under the assumption of point-closedness, we further conclude that the Sasaki maps are uniquely determined.

\begin{proposition} \label{prop:Sasaki-maps-are-Sasaki-projections}
	Any point-closed orthoset $(X,\perp)$ is a Sasaki space if and only if it is a Dacey space such that $\mathcal C(X,\perp)$ is AC, in which case for any $A \in {\mathcal C}(X,\perp)$, the Sasaki map to $A$ is given by
	\[ \{ \phi_{A}(e) \} \;=\; (\{ e \} \vee A\c) \cap A \]
	for any $e \notperp A$.
\end{proposition}

\begin{proof}
	The equivalence follows from Lemmas \ref{lem:pointclosed-completeatomistic}, \ref{lem:point closed Dacey with covering property implies Sasaki space} and \ref{lem:Sasaki-is-AC}, and Theorem \ref{thm:Sasaki-is-Dacey}.
	By Lemma \ref{lem:Sasaki-and-join-with-atom} we have that $A\c \vee \{e\} = A\c \vee \{ \phi_A(e) \}$. The assertion follows now from orthomodularity.
\end{proof}

For what follows, we need the characterization of orthomodular spacces by their associated subspace lattices. Furthermore, we call a lattice is \emph{irreducible} if it is not isomorphic to the direct product of two lattices with at least two elements. The proof of the following theorem follows from \cite[34.5]{MaMa}.
\begin{theorem} \label{thm:hermitian-space}
	Let $E$ be a orthomodular space. Then ${\mathcal C}(E)$, the ortholattice of orthoclosed subspaces of $E$, is a complete, irreducible, AC orthomodular lattice.
	
	Conversely, let $L$ be complete, irreducible, AC orthomodular lattice of height $\geq 4$. Then there is an orthomodular space $E$ such that $L$ is isomorphic to ${\mathcal C}(E)$.
\end{theorem}

Under the further assumption of irreducibility, we can now show that Sasaki spaces arise from orthomodular spaces.

\begin{theorem} \label{thm:Sasaki-spaces-Orthomodular-spaces}
	Let $(X,\perp)$ be an irreducible, point-closed Sasaki space of rank $\geq 4$. Then there is an orthomodular space $H$ such that $(X,\perp)$ is isomorphic to $(P(H),\perp)$.
	
	Under this isomorphism, the Sasaki maps are, for any closed subspace $S$ of $H$, given by
	\[ \phi_{P(S)}(\lin x) \;=\; (\lin x + S\c) \cap S, \]
	where $x \notperp S$.
\end{theorem}

\begin{proof}
	By Theorem \ref{thm:Sasaki-is-Dacey} and Lemma \ref{lem:Sasaki-is-AC}, ${\mathcal C}(X,\perp)$ is a complete AC orthomodular lattice of length $\geq 4$. Furthermore, as $(X,\perp)$ is irreducible, so is ${\mathcal C}(X,\perp)$ \cite[Lemma 3.6]{Vet2}. Hence, by Theorem \ref{thm:hermitian-space}, the first part follows. The second part is clear by Proposition \ref{prop:Sasaki-maps-are-Sasaki-projections}.
\end{proof}

\section{Sasaki projections}\label{sec:Sasaki projections}

In Theorem \ref{thm:Sasaki-is-Dacey}, we showed that the associated ortholattice $\mathcal C(X,\perp)$ of a Sasaki space $(X,\perp)$ is an orthomodular lattice. In  \cite{Fin}, Finch gave conditions on an orthocomplemented poset $P$ that imply that $P$ is an orthomodular lattice. Namely, he defined a \emph{Sasaki set of projections} on $P$ as a set $S$ of maps $P\to P$ such that
\begin{itemize}
	
	\item[\rm (i)] each $\pi\in S$ is monotone;
	\item[\rm (ii)] for each $\pi_1,\pi_2\in S$, we have that $\pi_1(1) \leq \pi_2(1)$ implies $\pi_1\circ\pi_2 = \pi_1$;
	\item[\rm (iii)] $\pi( \pi(x)\c ) \leq x\c$ for each $\pi\in S$ and each $x\in P$.
\end{itemize}
Furthermore, he called a Sasaki set of projection on $P$ \emph{full} if for each $x\in P$, there is some $\pi_x\in S$ such that $\pi_x(1)=x$, and showed that the existence of a full Sasaki set of projections on $P$ implies that $P$ is an orthomodular lattice, and that the projections in $S$ are actually Sasaki projections in the sense of Definition \ref{def:Sasaki projection}.  

In this section, let $(X,\perp)$ be a Sasaki space, which is neither necessarily point-closed nor irreducible. We show that the Sasaki maps of $(X,\perp)$ induce a full Sasaki set of projections on $\mathcal C(X,\perp)$. In order to do so, for any $A \in \linebreak {\mathcal C}(X,\perp)$, we choose a  Sasaki map $\phi_A$ to $A$ and we define
\[ \bar\phi_A \colon {\mathcal C}(X,\perp) \to {\mathcal C}(X,\perp)
\komma B \mapsto \{ \phi_A(e) \colon e \in B, \, e \notperp A \}\cc. \]
In the case that $(X,\perp)$ is point-closed, we may identify the elements of $X$ with the singletons and then we may view $\bar\phi_A$ as an extension of $\phi_A$. Indeed, for any $e \in X$ we have
\[ \bar\phi_A(\{e\}) \;=\;
\begin{cases} \{\phi_A(e)\} & \text{if $e \notperp A$,} \\
	\emptyset & \text{if $e \perp A$.}
\end{cases} \]

Furthermore, the next two lemmas show that $\bar\phi_A$ preserves arbitrary joins, and is thus determined by its action on singletons. 

\begin{lemma} \label{lem:selfadjointness-of-barphi}
	For any $A, B, C \in {\mathcal C}(X,\perp)$, we have
	\[ \bar\phi_A(B) \perp C \quad\text{if and only if}\quad B \perp \bar\phi_A(C). \]
\end{lemma}

\begin{proof}
	Assume that $\bar\phi_A(B) \perp C$. Then $\varphi_A(b)\perp c$ for each $b\in B\setminus A^\perp$ and each $c\in C$, so in particular for each $c\in C\setminus A^\perp$. Since $\varphi_A$ is a Sasaki map, it follows that $b\perp\varphi_A(c)$ for each $b\in B\setminus A^\perp$ and each $c\in C\setminus A^\perp$. Since $\varphi_A(c)\in A$ for each $c\in C\setminus A^\perp$, we have always have $b\perp \varphi_A(c)$ for each $b\in A^\perp$, hence  $b\perp\varphi_A(c)$ for each $b\in B$ and each $c\in C\setminus A^\perp$, which is equivalent to saying that $B\perp\bar\varphi_A(C)$. By symmetry, the assertion follows.
\end{proof}

\begin{lemma} \label{lem:barphi-is-joinpreserving}
	For any $A \in {\mathcal C}(X,\perp)$, $\bar\phi_A$ preserves arbitrary joins.
\end{lemma}

\begin{proof}
	Let $B_\iota \in {\mathcal C}(X,\perp)$, $\iota \in I$. By Lemma \ref{lem:selfadjointness-of-barphi}, we have for any $C \in {\mathcal C}(X,\perp)$ that $\bar\phi_A(\bigvee_\iota B_\iota) \perp C$ if and only if $\bigvee_\iota B_\iota \perp \bar\phi_A(C)$ if and only if $B_\iota \perp \bar\phi_A(C)$ for all $\iota$ if and only if $\bar\phi_A(B_\iota) \perp C$ for all $\iota$ if and only if $\bigvee_\iota \bar\phi_A(B_\iota) \perp C$. We conclude $\bar\phi_A(\bigvee_\iota B_\iota) = \bigvee_\iota \bar\phi_A(B_\iota)$.
\end{proof}

Note that the maps $\bar\phi_A$, $A \in {\mathcal C}(X,\perp)$, are in a one-to-one correspondence with ${\mathcal C}(X,\perp)$, as $\bar\phi_A(X) = A$. These maps determine the order of ${\mathcal C}(X,\perp)$ as follows.

\begin{lemma} \label{lem:barphi-determines-order}
	For any $A, B \in {\mathcal C}(X,\perp)$, $\bar\phi_A \circ \bar\phi_B = \bar\phi_A$ if and only if $\bar\phi_B \circ \bar\phi_A = \bar\phi_A$ if and only if $A \subseteq B$.
\end{lemma}

\begin{proof}
	Let $A, B \in {\mathcal C}(X,\perp)$ and assume that $\bar\phi_A \circ \bar\phi_B = \bar\phi_A$. For any $C, D \in {\mathcal C}(X,\perp)$, we then have, by Lemma \ref{lem:selfadjointness-of-barphi}, $\bar\phi_B(\bar\phi_A(C)) \perp D$ if and only if $C \perp \bar\phi_A(\bar\phi_B(D)) = \bar\phi_A(D)$ if and only if $\bar\phi_A(C) \perp D$. Hence we have $\bar\phi_B \circ \bar\phi_A = \bar\phi_A$.
	
	By similar reasoning, we conclude that $\bar\phi_B \circ \bar\phi_A = \bar\phi_A$ implies $\bar\phi_A \circ \bar\phi_B = \bar\phi_A$. Furthermore, from $\bar\phi_B \circ \bar\phi_A = \bar\phi_A$, it follows $A = \bar\phi_A(X) = \bar\phi_B(\bar\phi_A(X)) \subseteq B$. Finally, $A \subseteq B$ clearly implies $\bar\phi_B \circ \bar\phi_A = \bar\phi_A$.
\end{proof}

It is now evident that the maps $\bar\phi_A$, $A \in {\mathcal C}(X,\perp)$, are, in Finch's sense, a full Sasaki set of projections for the involutive poset ${\mathcal C}(X,\perp)$. Indeed, this means that the three properties in the following proposition hold.
\begin{proposition}
	For any $A, B \in {\mathcal C}(X,\perp)$, the following holds.
	\begin{itemize}
		
		\item[\rm (i)] $\bar\phi_A$ is order-preserving.
		
		\item[\rm (ii)] If $\bar\phi_A(X) \subseteq \bar\phi_B(X)$, then $\bar\phi_A \circ \bar\phi_B = \bar\phi_A$.
		
		\item[\rm (iii)] $\bar\phi_A( \bar\phi_A(B)\c ) \subseteq B\c$.
		
	\end{itemize}
\end{proposition}

\section{Classical Hilbert spaces}\label{sec:Classical Hilbert spaces}

We have characterised orthomodular spaces, sometimes also called generalised Hilbert spaces, as Sasaki spaces. We conclude by pointing out that, in the case of infinite rank, we may employ known conditions to characterise the Hilbert spaces over the real or complex number field. The tool is Sol\` er's theorem \cite{Sol}:

\begin{theorem} \label{thm:Soler}
	Let $H$ be an orthomodular space over a $\star$-field $K$ containing an infinite set of mutually orthogonal vectors of unit length. Then $K$ is one of $\Reals$, $\Complexes$, or $\Quaternions$, that is, $H$ is a classical Hilbert space.
\end{theorem}

Our postulate will again concern the existence of maps, which in this case, however, are symmetries. The following postulate was proposed, in similar form, e.g.~ in \cite{AeSt} and \cite{Vet}.

\begin{definition}
	We call an orthoset $(X,\perp)$ \emph{transitive} if, for any $e, f \in X$, there is an automorphism $\tau$ of $X$ such that $\tau(e) = f$ and $\tau(x) = x$ for any $x \perp e, f$.
\end{definition}

\begin{theorem}
	Let $(X,\perp)$ be a transitive, irreducible, point-closed Sasaki space of infinite rank. Then there is a Hilbert space $H$ over $\Reals$, $\Complexes$, or $\Quaternions$ such that $(X,\perp)$ is isomorphic to $(P(H),\perp)$.
\end{theorem}

\begin{proof}
	By Theorem \ref{thm:Sasaki-spaces-Orthomodular-spaces}, there is an orthomodular space $H$ over some $\star$-sfield $K$ such that $(X,\perp)$ is isomorphic to $(P(H),\perp)$. Without loss of generality, we may assume that $H$ has a unit vector $e_1$. As $X$ has infinite rank, $H$ is infinite-dimensional and hence possesses an infinite set $e_1, e_2,\ldots$ of mutually orthogonal vectors. By transitivity, there is, for each $i \geq 1$, an automorphism $\tau$ of $(P(H),\perp)$ such that $\tau(\lin{e_i}) = \lin{e_{i+1}}$ and $\tau(\lin x) = \lin x$ for any $x \perp e_i, e_{i+1}$. By Wigner's Theorem as formulated in \cite[Lemma 1]{May}, $\tau$ is induced by a unitary operator on $H$. We conclude that there is an infinite set of mutually orthogonal vectors of unit length. By Theorem \ref{thm:Soler}, the assertion follows.
\end{proof}

To single out the field of real or the one of complex numbers, also automorphisms can be used. We refer the reader to \cite{May}, see also \cite{Vet}.

\section*{Acknowledgements}
This research was supported by the bilateral Austrian Science Fund (FWF)
Project I 4579-N and Czech Science Foundation (GA\v{C}R) Project 20-09869L “The many facets of orthomodularity”.


\begin{thebibliography}{HHLN}

\bibitem[AmAr]{AmAr}
I. Amemiya, H. Araki, A remark on Piron's paper. \emph{Publ. Res. Inst. Math. Sci}. 2 (1966), no. 3, 423--427




\bibitem[AeSt]{AeSt} D. Aerts, B. Van Steirteghem,
Quantum axiomatics and a theorem of M. P. Sol\` er,
{\sl Int. J. Theor. Phys.} {\bf 39} (2000), 497--502.

\bibitem[BrHa]{BrHa} G. Bruns, J. Harding,
Algebraic aspects of orthomodular lattices,
in:
B. Coecke, D. Moore, A. Wilce,
``Current research in operational quantum logic. Algebras, categories, languages'',
Kluwer Academic Publishers, Dordrecht (2000);
37 - 65.

\bibitem[Dac]{Dac} J. R. Dacey,
``Orthomodular spaces'',
Ph.D. Thesis, University of Massachusetts, Amherst 1968.

\bibitem[EGL]{EGL} K. Engesser, D. M. Gabbay, D. Lehmann (Eds.),
``Handbook of Quantum Logic and Quantum Structures -- Quantum Structures'', 
Elsevier, Amsterdam 2007.

\bibitem[Fin]{Fin} P. D. Finch,
Sasaki projections on orthocomplemented posets,
{\it Bull. Aust. Math. Soc.} {\bf 1} (1969), 319 - 324.

\bibitem[Fou60]{Fou60} D. Foulis,
Baer *-semigroups,
{\it Proceedings of the American Mathematical Society} {\bf 11}, no. 4 (1960): 648 - 654.


\bibitem[Fou62]{Fou62} D. Foulis,
A note on Orthomodular Lattices,
{\it Portugaliae Mathematica} {\bf 21} Fasc. 1 (1962), 65 - 72.

\bibitem[HHLN]{HHLN}
J. Harding, C. Heunen, B. Lindenhovius, M. Navara, Boolean Subalgebras of Orthoalgebras, Order {\textbf 36} (2019), 563--609.


\bibitem[Kel]{Kel} H. A. Keller,
Ein nicht-klassischer Hilbertscher Raum,
{\it Math. Z.} {\bf 172} (1980), 41 - 49.

\bibitem[MaMa]{MaMa} F. Maeda, S. Maeda,
``Theory of symmetric lattices'',
Springer-Verlag, Berlin~-~Heidelberg~-~New York 1970.

\bibitem[May]{May} R. Mayet,
Some characterizations of the underlying division ring of a Hilbert lattice by automorphisms,
{\sl Int. J. Theor. Phys.} {\bf 37} (1998), 109 - 114.

\bibitem[Sol]{Sol} M. P. Sol\` er,
Characterization of Hilbert spaces by orthomodular spaces,
{\sl Commun. Algebra} {\bf 23} (1995), 219--243.

\bibitem[Vet19]{Vet} T. Vetterlein,
Orthogonality spaces of finite rank and the complex Hilbert spaces,
{\it Int. J. Geom. Methods Mod. Phys.} {\bf 16} No.~05 (2019), 1950080.

\bibitem[Vet21]{Vet2} T. Vetterlein, Gradual transitivity in orthogonality spaces of finite rank
{\it Aequationes mathematicae} {\bf  95} (2019), 483--503.


\bibitem[Vet22]{Vet3} T. Vetterlein, Transitivity and homogeneity of orthosets and inner-product spaces over subfields of $\mathbb{R}$, \emph{Geometriae Dedicata} {\bf  216} (2022), Article number: 36.



\bibitem[Wil]{Wil} A. Wilce,
Test spaces,
in \cite{EGL}; 443 - 549.

\end{thebibliography}
\end{document}